\documentclass[runningheads,a4paper]{article}
\usepackage[latin1]{inputenc}
\usepackage{graphicx}
\usepackage{color}
\usepackage{xspace}
\usepackage{algorithm2e}
\usepackage{amsmath}
\usepackage{epstopdf}
\usepackage{amssymb}
\usepackage{amsthm}
\usepackage{subfig}
\usepackage{algpseudocode}

\newtheorem{theorem}{Theorem	}
\newtheorem{definition}{Definition}

\newtheorem{mylemma}{Lemma}
\newtheorem{corollary}{Corollary}

\begin{document}

\newlength{\figurewidth}
\setlength{\figurewidth}{0.7\textwidth}

\newcommand{\argmin}{\mathop{\mathrm{argmin}}} 
%\renewcommand{\qedsymbol}{$\blacksquare$}

% Paper limit: 12 pages

% Correct bad hyphenation here
\hyphenation{net-works semi-conduc-tor me-di-cal le-ve-rage acknow-ledge
acknow-ledgment pro-blems con-fi-gu-ra-tions}

% Title
\title{Every Schnyder Drawing is a Greedy Embedding}

% Authors
\author{ Pierre Leone and Kasun Samarasinghe\\ Centre Universitaire d'Informatique \\ University of Geneva\\ Switzerland\\ \{pierre.leone,kasun.wijesiriwardana\}@unige.ch}

% Make title
\date{}
\maketitle

% Abstract
\begin{abstract}\noindent
Geographic routing is a routing paradigm, which uses geographic coordinates of network nodes to determine routes. Greedy routing, the simplest form of geographic routing forwards a packet to the closest neighbor towards the destination. A greedy embedding is a embedding of a graph on a geometric space such that greedy routing always guarantees delivery. A Schnyder drawing is a classical way to draw a planar graph. In this manuscript, we show that every Schnyder drawing is a greedy embedding, based on a generalized definition of greedy routing.
\end{abstract}

% Introduction
\section{Greedy Routing on Planar Triangular Graphs}
In this manuscript, we establish some results on greedy routing on planar triangular graphs. Planar triangular graphs are a special class of graphs, where every face of the graph is a triangle, including the external face. Hence the outer-face consists of three nodes $A_{1},A_{2}$ and $A_{3}$.

We consider the problem of greedy routing~\cite{bose1999online} on such graphs. In order to perform geometric routing, the graph has to be drawn on a certain geometric space. Such a drawing is called a greedy embedding~\cite{papadimitriou2004conjecture}, which is defined as follows.

\begin{definition}\textbf{Greedy Embedding}
A greedy embedding is an embedding of a graph on a respective geometric space such that, greedy routing always succeeds. In other words, between every node pair $u,v$ there is another node $w$ adjacent to $u$, such that $d(u,v)>d(w,v)$, where $d(.)$ is the underlying metric on the geometric space.
\end{definition}

\noindent Dhandapani~\cite{dhandapani2010greedy} showed that every planar triangulated graph can be drawn on the plane as a greedy embedding. They generalize the classical Schnyder drawing~\cite{schnyder1990embedding} leading to  a family of planar drawings, then they show that there exists a greedy drawing in this set of drawings.

In this work, we prove that every Schnyder drawing is a greedy embedding. We emphasize the use of a generalized definition of a greedy routing~\cite{li2010rules}, on which our algorithm is based. 

\section{Schnyder Drawing}
Given a planar triangular graph, a Schnyder drawing~\cite{schnyder1990embedding} is a straight line drawing of the graph on the plane. In this article, we consider such a drawing on $\mathbb{R}^3$, such that the external nodes $A_{1},A_{2}$ and $A_{3}$ are placed on $(1,0,0), (0,1,0)$ and $(0,0,1)$ respectively. Hence the external face forms an equilateral triangle and the nodes are placed on the plane designated by $x+y+z=1$. 

The Schnyder drawing is computed based on a combinatorial description of a planar triangular graph, which is called a \textit{realizer}, defined as follows.

\begin{theorem}\textbf{Realizer}\cite{schnyder1990embedding}
Given a plane triangulation $G(V,E)$, there exist three directed edge-disjoint trees, $T_{1},T_{2}$ and $T_{3}$, namely the realizer of G, such that for each inner vertex u;
\begin{enumerate}
\item u has an outgoing edge in each of $T_{1},T_{2}$ and $T_{3}$
\item the counterclockwise order of the edges incident on v is as follows: leaving in $T_{1}$, entering in $T_{3}$, leaving in $T_{2}$, entering $T_{1}$, leaving in $T_{3}$, entering in $T_{2}$
\end{enumerate}
\end{theorem}

\noindent The Schnyder planar drawing algorithm~\cite{schnyder1990embedding}, initially constructs a realizer in linear time. A realizer, in turn leads to three paths from each node towards their root nodes in each tree. These paths partition the nodes into three regions $R_{i}$ such that $i={1,2,3}$. Let $n_{i}$ be the total number of nodes in region $R_{i}$ including the nodes in the two paths, those border the region $R_{i}$. Now Schnyder algorithm places each node $u$ on $\frac{1}{n}(n_{1},n_{2},n_{3})$ leading to a planar drawing (see~\cite{nishizeki2004planar} for details)~\cite{schnyder1990embedding}. \footnote{Note that the drawing we consider forms an equilateral triangle, while in~\cite{nishizeki2004planar} they present an algorithm where the outer face does not form an equilateral triangle} 

We present two important properties of a Schnyder drawing in Lemma~\ref{threewedge} and~\ref{lem:enclosing}, which are illustrated in Figures~\ref{fig:threewedge} and~\ref{fig:enclosing}\footnote{See~\cite{dhandapani2010greedy} for the proof}

\begin{figure}
  \centering
  \subfloat[Illustration of Lemma~\ref{threewedge}: There are exactly three edges in the three regions $S_1^u$,$S_3^u$ and $S_5^u$]{\label{fig:threewedge}\includegraphics[width=0.45\textwidth]{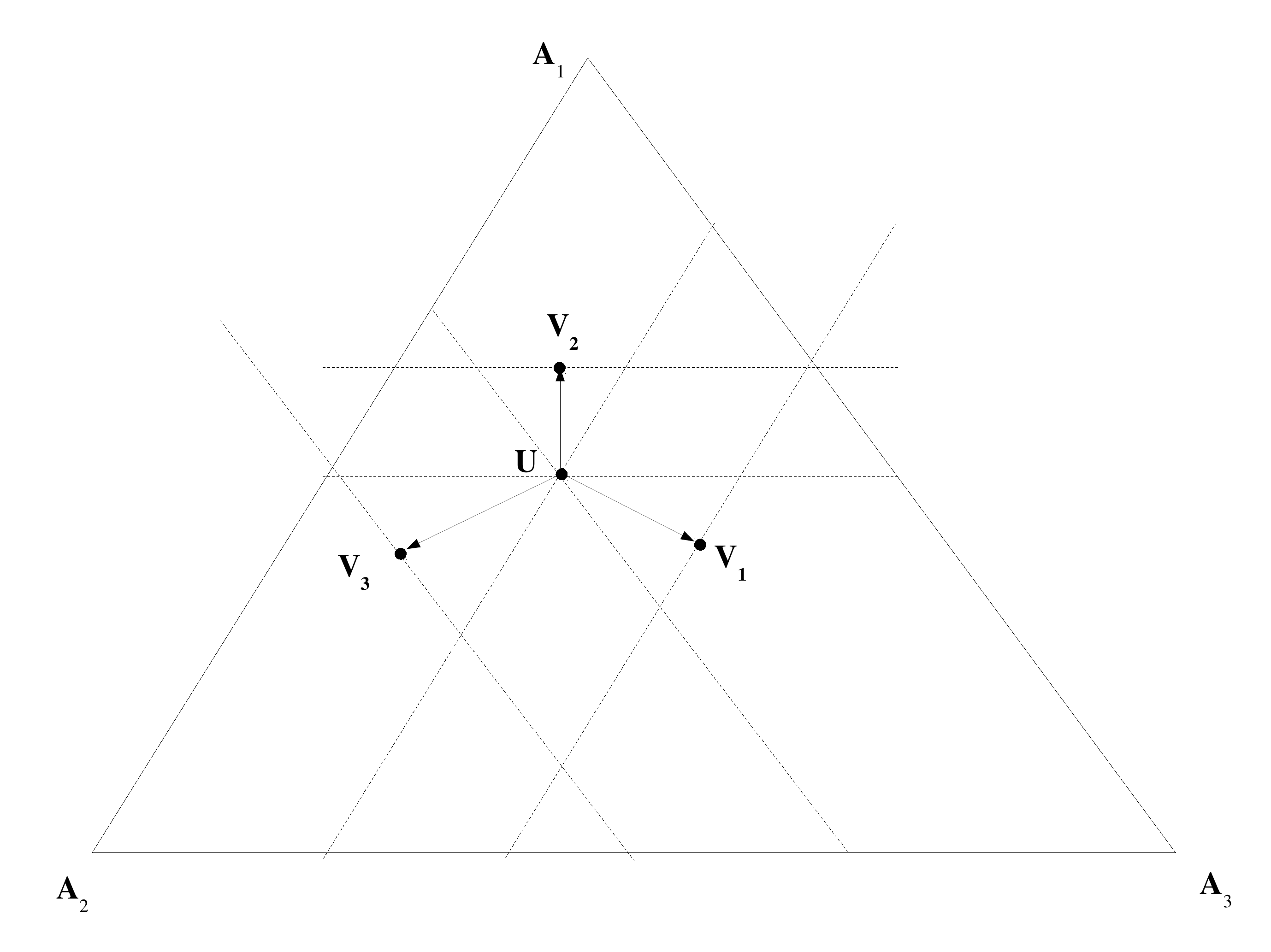}}           \hfill
  \subfloat[Gray region is empty of nodes; Enclosing triangle property (Lemma 6)~\cite{dhandapani2010greedy} ]{\label{fig:enclosing}\includegraphics[width=0.45\textwidth]{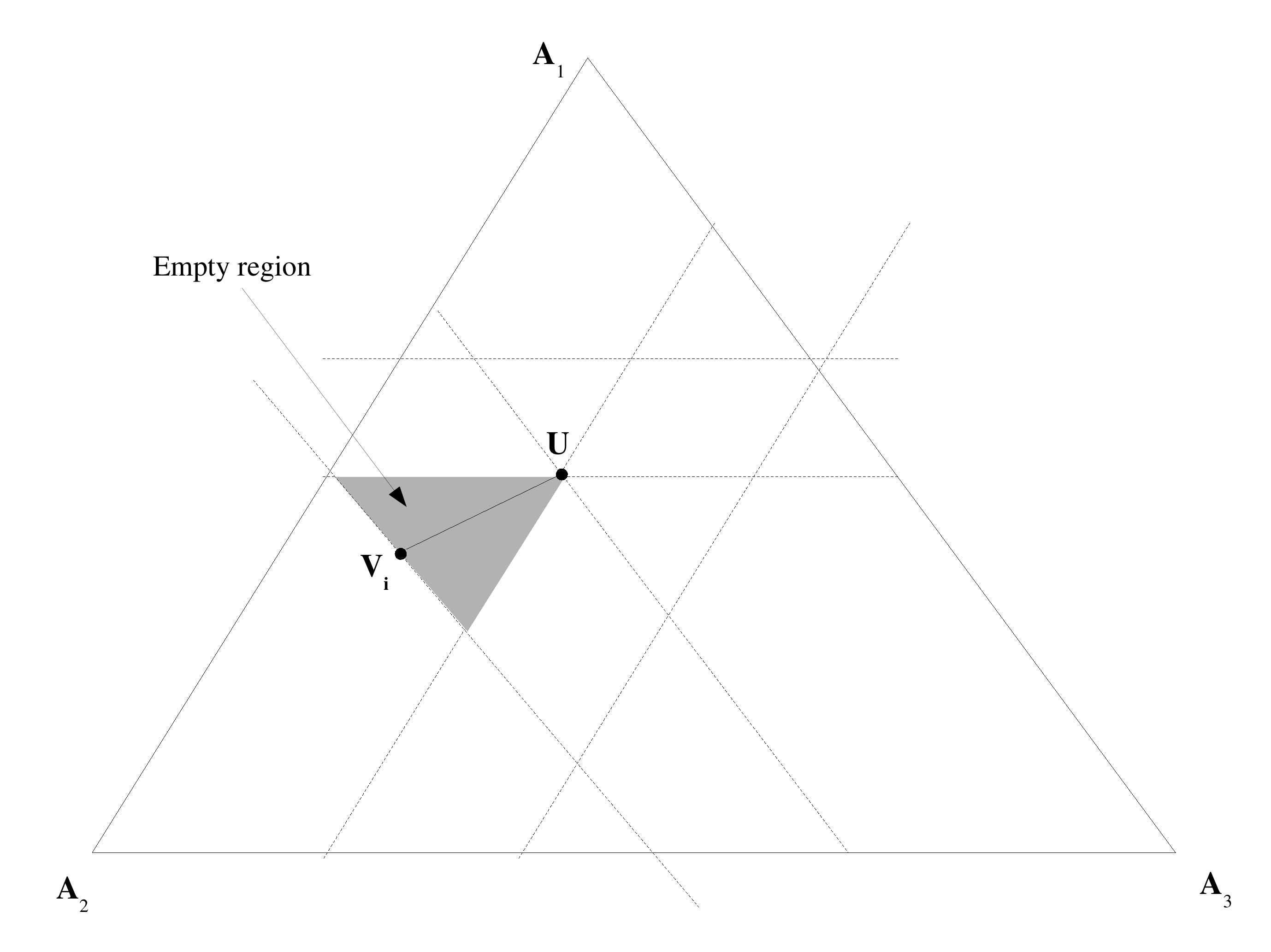}}
  \caption{Two cases to consider in greedy path construction}
  \label{fig:overview}
\vspace{-0.2cm}
\end{figure}	

\begin{mylemma}\textbf{The Three Wedges Property} Lemma 4\cite{dhandapani2010greedy}
\label{threewedge}\\
In every Schnyder drawing the three outgoing edges at an internal vertex $v$ have slopes that fall in the intervals $P_{1}$ in $[60^{o},120^{o}]$, $P_{2}$ in $[180^{o},240^{o}]$ and $P_{3}$ in $[300^{o},360^{o}]$, with exactly one edge in each interval as shown in Figure~\ref{fig:threewedge}
\end{mylemma}

\noindent Schnyder drawing leads to an important void region around and out-going edge as constitutes in Lemma~\ref{lem:enclosing}

\begin{mylemma}\textbf{The enclosing triangle property} Lemma 6\cite{dhandapani2010greedy}
Given a vertex u and an outgoing edge (u,v) belonging, without loss of generality to $P_{1}$, the equilateral triangle formed by drawing lines with slopes of $0^{o},60^{o}$ and $120^{o}$, as shown in Figure~\ref{fig:enclosing} is free of any other vertices. This results holds for both regions $P_{2}$ and $P_{3}$.
\label{lem:enclosing}
\end{mylemma}

\noindent In the following section, we introduce the definition of a saturated graph and show that every Schnyder drawing implies a saturated graph.

\section{Saturated Graph}
In~\cite{leone2016greedy}, we introduced the concept of saturated graphs and show that there exists a local greedy routing algorithm with guaranteed delivery on such graphs. Definition of a saturated graph is purely combinatorial, but to construct those properties we use an underlying virtual coordinate system namely, virtual raw anchor coordinate (VRAC) system~\cite{leone2016greedy}. In the following, we describe the construction of a saturated graph, given a Schnyder drawing of a planar triangular graph.

Let $G(V,E)$ is the graph in concern, where $V$ is the set of vertices(nodes) and $E$ is the set of edges, which is a planar triangular graph. Consider a Schnyder drawing of $G$, where the outer face is drawn as an equilateral triangle. We denote the three outer vertices as $A_{1},A_{2}$ and $A_{3}$. Let every node computes the distance from the edges of the outer-triangle and assign them as its coordinate as illustrated in Figure~\ref{fig:coordinates}, which is in fact the VRAC system. Based on this coordinate assignment, we can define three order relations on the set of nodes $V$ as follows. 

\begin{figure}
  \centering
  \includegraphics[width=0.6\textwidth]{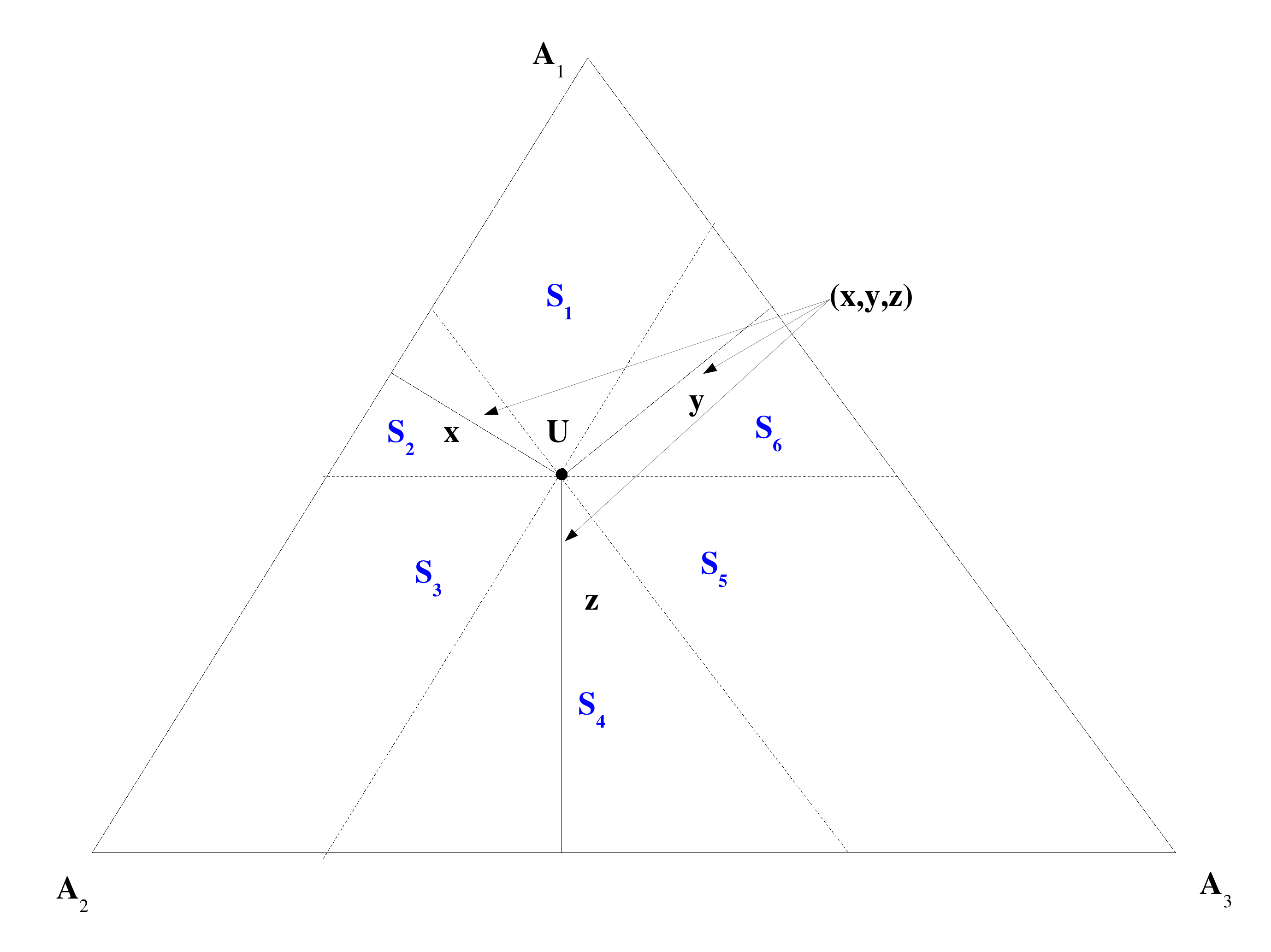}  
  \caption{A node estimates the perpendicular distances from the outer-face edges and assign itself the raw distances as the coordinate (x,y,z)}
  \label{fig:coordinates}
\end{figure}

\begin{definition}\label{def:order} The three order relations $<_i$, $i=1,2,3$ on $V\times V$ are defined by
\begin{equation*}
\forall u,v\in V\quad u<_i v  \Longleftrightarrow d(u,A_i)>d(v,A_i)   \Longleftrightarrow u_i < v_i .
\end{equation*}
\end{definition}

\noindent These three orders permit the definition of sectors associated with a node $u$.\\

\begin{definition}\label{def:sectors} We define the following sectors associated with a node $u\in V$. Note that the reference node $u$ does not belong either of to the sectors. As depicted in Figure~\ref{fig:coordinates}, sectors $s_{i}^u$ correspond to the labeled regions.
\begin{align*}
s_1^u = \{v\mid ~u~<_ 1 ~v, ~u~>_ 2 ~v, ~u~>_ 3 ~v\}~~ \cap~~ \widehat{A_1A_2A_3}.\\
s_2^u = \{v\mid ~u~<_ 1 ~v, ~u~<_ 2 ~v, ~u~>_ 3 ~v\}~~ \cap~~ \widehat{A_1A_2A_3}.\\
s_3^u = \{v\mid ~u~>_ 1 ~v, ~u~<_ 2 ~v, ~u~>_ 3 ~v\}~~ \cap~~ \widehat{A_1A_2A_3}.\\
s_4^u = \{v\mid ~u~>_ 1 ~v, ~u~<_ 2 ~v, ~u~<_ 3 ~v\}~~ \cap~~ \widehat{A_1A_2A_3}.\\
s_5^u = \{v\mid ~u~>_ 1 ~v, ~u~>_ 2 ~v, ~u~<_ 3 ~v\}~~ \cap~~ \widehat{A_1A_2A_3}.\\
s_6^u = \{v\mid ~u~<_ 1 ~v, ~u~>_ 2 ~v, ~u~<_ 3 ~v\}~~ \cap~~ \widehat{A_1A_2A_3}.\\
\end{align*}
\end{definition}

\noindent Based on this characterization, Leone et.al~\cite{huc2012efficient}, proposed a distributed planarization algorithm, assuming a unit disk communication graph model. They used the following planarization criterion, where they derive a planar sub-graph $G(V,\widetilde E)$.

\begin{definition}\label{def:planar} The vertex set $\widetilde V = V$ and
\begin{eqnarray}
\widetilde E = \biggl\{ (u,v) \Big\vert~v\in s_{2k-1}^u \text{ and } v=\text{ min}_k (s_{2k-1}^u)~k=1,2\text{ or }3 \text{ and } (u,v)\in E\biggr\}
\label{minprop}
\end{eqnarray}
\end{definition}

\noindent Algorithmically, to obtain a planar graph, each node retains only the minimum edge in each sector $S_1^u,S_3^u$ and $S_5^u$. The minimum edge in a given sector is determined based on a partial order corresponds for each sector. Such a partial order is defined considering the intersection of three orders ($<_i$, $i=1,2,3$) in a given sectors (see~\cite{huc2012efficient} for details). Note that given an arbitrary graph, these three sectors may not have edges and there can be arbitrary number of incoming edges to a node in sectors $S_2^u,S_4^u$ and $S_6^u$. A saturated graph is a special case of this setting, which is defined as below.

\begin{definition}[Saturated Graph]\label{saturated} A planar graph is saturated if there exists exactly one edge in each sector $s^u_{2i-1}, i=1,2,3$ for each node $u$.
\end{definition}

\noindent We present the following result from~\cite{leone2016greedy} on greedy routing on saturated graphs.

\begin{theorem}
There is a greedy routing algorithm on every saturated planar graph. 
\end{theorem}

\subsection{Schnyder drawings and saturated graphs}

A saturated graph is defined without a reference to an embedding. Our greedy routing algorithm in~\cite{leone2016greedy}, uses the saturated graph property to prove the delivery guarantees. Following lemma constitutes a straight forward relationship between a Schnyder drawing and a saturated graph.

\begin{mylemma}
Every Schnyder drawing implies a saturated graph.
\end{mylemma}

\begin{proof}
Due to the three wedge property of a Schnyder drawing (see Lemma~\ref{threewedge}), we know that there is exactly one edge $(u,v_{i})$ in sectors $s_i^u$ where $i={1,3,5}$. Moreover due to the $enclosing\ triangle\ property$(see Lemma~\ref{lem:enclosing}) of a Schnyder drawing, there is no node $w$ such that $w<_{j}v$ where $j={1,2,3}$. Hence it follows the criterion for a saturated edge as in equation~\ref{minprop}, implying a saturated graph.    
\end{proof}

\noindent In~\cite{leone2016greedy}, we devised a greedy routing algorithm which guarantees delivery, when the graph is saturated. Note that in~\cite{leone2016greedy}, we do not use a metric to define the greedy path, instead use a generalized definition of a greedy path. Following lemma concludes the resulting connection between greedy embeddings and a Schnyder drawing of a planar triangular graph.

\begin{corollary}
Every Schnyder drawing is a greedy embedding.
\end{corollary}

% Related work

\bibliographystyle{unsrt}
\bibliography{/home/kasun/Unige/PhD/Bib/adhoc-routing}
% Bibliography
\end{document}